\documentclass[11pt]{article}
\usepackage[margin=1in]{geometry}

\usepackage{amssymb,amsmath,amsthm}
\usepackage{paralist}
\usepackage{color}

\newtheorem{theorem}{Theorem}
\newtheorem{prop}[theorem]{Proposition}
\newtheorem{lemma}[theorem]{Lemma}

\def\eps{\epsilon}
\def\R{\mathbb{R}}
\def\suchthat{\;:\;}

\newcommand{\abs}[1]{\left|#1\right|}

\newcommand{\norm}[1]{\left\|#1\right\|}
\newcommand{\sqd}[1]{\left\|#1\right\|^{2}}
\newcommand{\prob}[1]{\operatorname{Pr}\left(#1\right)}
\newcommand{\expec}[1]{\operatorname{E}\left[#1\right]}
\newcommand{\size}[1]{\left|#1\right|}
\newcommand{\linspan}[1]{\operatorname{span}\left(#1\right)}
\newcommand{\inner}[2]{\left\langle #1, #2\right\rangle}

\title{Subspace approximation with outliers}

\author{Amit Deshpande\footnote{Microsoft Research India. \texttt{amitdesh@microsoft.com}} \and Rameshwar Pratap\footnote{IIT Mandi \texttt{rameshwar.pratap@gmail.com}}}

\begin{document}

\maketitle

\begin{abstract}
The subspace approximation problem with outliers, for given $n$ points in $d$ dimensions $x_{1}, x_{2}, \dotsc, x_{n} \in \R^{d}$, an integer $1 \leq k \leq d$, and an outlier parameter $0 \leq \alpha \leq 1$, is to find a $k$-dimensional linear subspace of $\R^{d}$ that minimizes the sum of squared distances to its nearest $(1-\alpha)n$ points. More generally, the $\ell_{p}$ subspace approximation problem with outliers minimizes the sum of $p$-th powers of distances instead of the sum of squared distances. Even the case of $p=2$ or robust PCA is non-trivial, and previous work requires additional assumptions on the input or generative models for it. 
Any multiplicative approximation algorithm for the subspace approximation problem with outliers must solve the robust subspace recovery problem, a special case in which the $(1-\alpha)n$ inliers in the optimal solution are promised to lie exactly on a $k$-dimensional linear subspace. However, robust subspace recovery is Small Set Expansion (SSE)-hard, and known algorithmic results for robust subspace recovery require strong assumptions on the input, e.g., any $d$ outliers must be linearly independent.

In this paper, we show how to extend dimension reduction techniques and bi-criteria approximations based on sampling and coresets to the problem of subspace approximation with outliers. To get around the SSE-hardness of robust subspace recovery, we assume that the squared distance error of the optimal $k$-dimensional subspace summed over the optimal $(1-\alpha)n$ inliers is at least $\delta$ times its squared-error summed over all $n$ points, for some $0 < \delta \leq 1 - \alpha$. Under this assumption, we give an efficient algorithm to find a \emph{weak coreset} or a subset of $\text{poly}(k/\eps) \log(1/\delta) \log\log(1/\delta)$ points whose span contains a $k$-dimensional subspace that gives a multiplicative $(1+\eps)$-approximation to the optimal solution. 
%The entire span of this subset is a linear subspace of dimension $\text{poly}(k/\eps) \log(1/\delta) \log\log(1/\delta)$ that gives a bi-criteria $(1+\eps)$-approximation to the optimal $k$-dimensional solution. 
The running time of our algorithm is linear in $n$ and $d$. Interestingly, our results hold even when the fraction of outliers $\alpha$ is large, as long as the obvious condition $0 < \delta \leq 1 - \alpha$ is satisfied.
We show similar results for subspace approximation with $\ell_{p}$ error or more general M-estimator loss functions, and also give an additive approximation for the affine subspace approximation problem.
\end{abstract}

\section{Introduction} \label{sec:intro}
Finding low-dimensional representations of large, high-dimensional input data is an important first step for several problems in computational geometry, data mining, machine learning, and statistics. For given input points $x_{1}, x_{2}, \dotsc, x_{n} \in \R^{d}$, a positive integer $1 \leq k \leq d$ (typically much smaller than $d$) and $1 \leq p < \infty$, the $\ell_{p}$ subspace approximation problem asks to find a $k$-dimensional linear subspace $V$ of $\R^{d}$ that essentially minimizes the sum of $p$-th powers of the distances of all the points to the subspace $V$, or to be precise, it minimizes the $\ell_{p}$ error
\[
\left(\sum_{i=1}^{n} d(x_{i}, V)^{p}\right)^{1/p} \quad \text{or equivalently} \quad \sum_{i=1}^{n} d(x_{i}, V)^{p}.
\]
For $p=2$, the optimal subspace is spanned by the top $k$ right singular vectors of the matrix $X \in \R^{n \times d}$ formed by $x_{1}, x_{2}, \dotsc, x_{n}$ as its rows. The optimal solution for $p=2$ can be computed efficiently by the Singular Value Decomposition (SVD) in time $O\left(\min\{nd^{2}, n^{2}d\}\right)$. Liberty's deterministic matrix sketching \cite{Liberty13} and subsequent work \cite{GhashamiP14} provide a faster, deterministic algorithm that runs in $O\left(nd \cdot \text{poly}(k/\eps)\right)$ time and gives a multiplicative $(1+\eps)$-approximation to the optimum. There is also a long line of work on randomized algorithms \cite{Sarlos06,Woodruff-monograph} that sample a subset of points and output a subspace from their span, giving a multiplicative $(1+\eps)$-approximation in running time $O(\text{nnz}(X)) + (n+d) \text{poly}(k/\eps)$, where $\text{nnz}(X)$ is the number of non-zero entries in $X$. These are especially useful on sparse data.

For $p \neq 2$, unlike the $p = 2$ case, we do not know any simple description of the optimal subspace. For any $p \geq 1$, Shyamalkumar and Varadarajan \cite{ShyamalkumarV12} give a $(1+\eps)$-approximation algorithm that runs in time $O\left(nd \cdot \exp((k/\eps)^{O(p)})\right)$. Building upon this, Deshpande and Varadarajan \cite{DeshpandeV07} give a bi-criteria $(1+\eps)$-approximation by finding a subset of $s = (k/\eps)^{O(p)}$ points in time $O\left(nd \cdot \text{poly}(k/\eps)\right)$ such that their $s$-dimensional linear span gives a $(1+\eps)$-approximation to the optimal $k$-dimensional subspace. The subset they find is basically a \emph{weak coreset}, and projecting onto its span also  gives dimension-reduction result for subspace approximation. Feldman et al. \cite{FeldmanMSW10} improve the running time to $nd \cdot \text{poly}(k/\eps) + (n+d) \cdot \exp(\text{poly}(k/\eps))$ for $p=1$. Feldman and Langberg \cite{FeldmanL11} extend this result to achieve a running time of $nd \cdot \text{poly}(k/\eps) + \exp((k/\eps)^{O(p)})$ for any $p \geq 1$. Clarkson and Woodruff \cite{ClarksonW15} improve this running time to $O(\text{nnz}(X) + (n+d) \cdot \text{poly}(k/\eps) + \exp(\text{poly}(k/\eps)))$ for any $p \in [1, 2)$. The case $p \in [1, 2)$, especially $p=1$, is important because the $\ell_{1}$ error (i.e., the sum of distances) is more robust to outliers than the $\ell_{2}$ error (i.e., the sum of squared distances). 

We consider the following variant of $\ell_{p}$ subspace approximation in the presence of outliers. Given points $x_{1}, x_{2}, \dotsc, x_{n} \in \R^{d}$, an integer $1 \leq k \leq d$, $1 \leq p < \infty$, and an outlier parameter $0 \leq \alpha \leq 1$, find a $k$-dimensional linear subspace $V$ that minimizes the sum of $p$-th powers of distances of the $(1-\alpha)n$ points nearest to it. 
In other words, let $N_{\alpha}(V) \subseteq [n]$ consist of the indices of the nearest $(1-\alpha)n$ points to $V$ among $x_{1}, x_{2}, \dotsc, x_{n}$. We want to minimize
$
\sum_{i \in N_{\alpha}(V)} d(x_{i}, V)^{p}.
$

The robust subspace recovery problem is a special case in which the optimal error for the subspace approximation problem with outliers is promised to be zero, that is, the optimal subspace $V$ is promised to go through some $(1-\alpha)n$ points among $x_{1}, x_{2}, \dotsc, x_{n}$. Thus, any multiplicative approximation must also have zero error and recover the optimal subspace. Khachiyan \cite{Khachiyan1995} proved that it is NP-hard to find a $(d-1)$-dimensional subspace that contains at least $(1-\eps)(1 - 1/d)n$ points. Hardt and Moitra \cite{HardtM13} study robust subspace recovery and define an $(\eps, \delta)$-Gap-Inlier problem of distinguishing between these two cases: \begin{inparaenum}[(a)] \item there exists a subspace of dimension $\delta n$ containing $(1-\eps)\delta n$ points and \item every subspace of dimension $\delta n$ contains at most $\eps \delta n$ points. \end{inparaenum} They show a polynomial time reduction from the $(\eps, \delta)$-Gap-Small-Subset-Expansion problem to the $(\eps, \delta)$-Gap-Inlier problem. For more on Small Set Expansion conjecture and its connections to Unique Games, please see \cite{RS2010}. Under a strong assumption on the data (that requires any $d$ or fewer outliers to be linearly independent), Hardt and Moitra give an efficient algorithms for finding $ k$-dimensional subspace containing $(1 - k/d) n$ points. This naturally leaves open the question of finding other more reasonable approximations to the subspace approximation problem with outliers.

In recent independent work, Bhaskara and Kumar (see Theorem 12 in \cite{BK2018}) showed that if $(\eps, \delta)$-Gap-Small-Subset-Expansion problem is NP-hard, then there exists an instance of subspace approximation with outliers where the optimal inliers lie on a $k$-dimensional subspace but it is NP-hard to find even a subspace of dimension $O(k/\sqrt{\eps})$ that contains all but $(1 + \delta/4)$ times more points than the optimal number of outliers. This showed that even bi-criteria approximation for subspace recovery is a challenging problem. We compare and contrast our results with the result of Bhaskara and Kumar \cite{BK2018}. Their algorithm throws more outlier than the optimal solution, while we don't throw any extra outlier. Also, their bi-criteria approximation depends on the ``rank-$k$ condition'' number which is a somewhat stronger assumption than ours. %Bhaskara and Kumar \cite{BK2018} show weaker results than ours under a somewhat stronger assumption on what they call as rank-$k$ condition number of the input data.

The problem of clustering using points and lines in the presence of outliers has been studied in special cases of $k$-median and $k$-means clustering \cite{Chen2008,KrishnaswamyLS2018}, and points and line clustering \cite{FeldmanS12}. Krishnaswamy et al. \cite{KrishnaswamyLS2018} give a constant factor approximation for $k$-median and $k$-means clustering with outliers, whereas Feldman and Schulman give $(1+\eps)$-approximations for $k$-median with outliers and $k$-line median with outliers that run in time linear in $n$ and $d$.

Another recent line of research on robust regression considers data coming from an underlying distribution where a fraction of it is arbitrarily corrupted \cite{LRV2016,KKM2018}. The problem we study is different as we do not assume any generative model for the input.

\section{Our contributions}
\begin{itemize}
\item We assume that the $\ell_{p}$ error of the optimal subspace summed over the optimal $(1-\alpha)n$ inliers is at least $\delta$ times its total $\ell_{p}$ error summed over all $n$ points, for some $\delta > 0$. Under this assumption, we give an algorithm to efficiently find a subset of $\text{poly}(pk/\eps) \cdot \log(1/\delta) \log\log(1/\delta)$ points from $x_{1}, x_{2}, \dotsc, x_{n}$ such that the span of this subset contains a $k$-dimensional linear subspace whose $\ell_{p}$ error over its nearest $(1-\alpha)n$ points is within $(1+\eps)$ of the optimum. The running time of our algorithm is linear in $n$ and $d$. Note that even for $\delta$ as small as $1/\text{poly}(n)$, our algorithm outputs a fairly small subset of $\text{poly}(pk/\eps) \cdot \log n \log\log n$ points. The running time of our sampling-based algorithm is linear in $n$ and $d$.

\item Alternatively, the entire span of the above subset is a linear subspace of dimension $\text{poly}(pk/\eps) \cdot \log(1/\delta) \log\log(1/\delta)$ that gives a bi-criteria multiplicative $(1+\eps)$-approximation to the optimal $k$-dimensional solution to the $\ell_{p}$ subspace approximation problem with outliers. Interestingly, this holds even when the fraction of outliers $\alpha$ is large, as long as the obvious condition $0 < \delta \leq 1 - \alpha$ is satisfied.

\item Our assumption that the $\ell_{p}$ error of the optimal subspace summed over the optimal $(1-\alpha)n$ inliers is at least $\delta$ times its total $\ell_{p}$ error summed over all $n$ points, for some $\delta > 0$, is more reasonable and realistic than the assumptions used in previous work on subspace approximation with outliers. Without this assumption, our problem (even its special case of subspace recovery) is known to be Small Set Expansion (SSE)-hard \cite{HardtM13}.

\item The technical contribution of our work is in showing that the sampling-based weak coreset constructions and dimension reduction results for the subspace approximation problem without outliers \cite{DeshpandeV07} also extend to its robust version for data with outliers. 
If we know the inlier-outlier partition of the data, then   the  result of~\cite{DeshpandeV07} can easily be extended for the outlier version of the problem. However, if we don't know such a partitioning, then a brute-force approach has to go over all $n \choose (1-\alpha)n$ subsets and picks the best solution. This is certainly not an efficient approach as the number of such subsets is exponential in $n$.
Further, on inputs that satisfy our assumption (stated above), it is easy to see that solving the subspace approximation problem without outliers gives a multiplicative $1/\delta$-approximation to the subspace approximation problem with outliers. Our contribution lies in showing that this approximation guarantee can be improved significantly in a small number of additional sampling steps.

\item We show immediate extensions of our results to more general M-estimator loss functions as previously considered by \cite{ClarksonW15}. 

\item We show that our multiplicative approximation for the \emph{linear} subspace approximation problem under $\ell_{2}$ error implies an additive approximation for the \emph{affine} subspace approximation problem under $\ell_{2}$ error. The running time of this algorithm is also linear in $n$ and $d$.
\end{itemize}

\section{Warm-up: least squared error line approximation with outliers}
%\vspace{-2mm}
As a warm-up towards the main proof, we first consider the case $k=1$ and $p=2$, that is, for a given $0 < \alpha < 1$, we want to find the best line that minimizes the sum of squared distances summed over its nearest $(1-\alpha)n$ points. Let $x_{1}, x_{2}, \dotsc, x_{n} \in \R^{d}$ be the given points and let $l^{*}$ be the optimal line. Let $I \subseteq [n]$ consist of the indices of the nearest $(1-\alpha)n$ points to $l^{*}$ among $x_{1}, x_{2}, \dotsc, x_{n}$.

Our algorithm iteratively builds a subset $S \subseteq [n]$ by starting from $S = \emptyset$ and in each step samples with replacement $\text{poly}(k/\eps)$ i.i.d. points where each point $x_{i}$ is picked with probability proportional to its squared distance to the span of the current subset $d(x_{i}, \linspan{S})^{2}$. We abuse the notation as $\linspan{S}$ to denote the linear subspace spanned by $\{x_{i} \suchthat i \in S\}$. These sampled points are added to $S$ and the sampling algorithm is repeated $\text{poly}(k/\eps)$ times. 

\subsection{Additive approximation}
We are looking for a small subset $S \subseteq [n]$ of size $\text{poly}(1/\eps)$ that contains a close additive approximation to the optimal subspace over the optimal inliers, that is, for the projection of $l^{*}$ onto $\linspan{S}$ denoted by $P_{S}(l^{*})$,
\begin{align}
\sum_{i \in I} d(x_{i}, P_{S}(l^{*}))^{2} \leq \sum_{i \in I} d(x_{i}, l^{*})^{2} + \eps~ \sum_{i=1}^{n} \sqd{x_{i}} \label{eq:add-approx}
\end{align}
This immediately implies that there exists a line $l_{S}$ in $\linspan{S}$ such that 
\[
\sum_{i \in N_{\alpha}(l_{S})} d(x_{i}, l_{S})^{2} \leq \sum_{i \in I} d(x_{i}, l^{*})^{2} + \eps~ \sum_{i=1}^{n} \sqd{x_{i}},
\]
where $N_{\alpha}(l_{S}) \subseteq [n]$ consists of the indices of the nearest $(1-\alpha)n$ points from $x_{1}, x_{2}, \dotsc, x_{n}$ to $l_{S}$.

Given any subset $S \subseteq [n]$, define the set of \emph{bad} points as a subset of inliers $I$ whose error w.r.t. $P_{S}(l^{*})$ is somewhat larger than their error w.r.t. the optimal line $l^{*}$, that is, $B(S) = \{i \in I \suchthat d(x_{i}, P_{S}(l^{*}))^{2} > (1+\eps/2)~ d(x_{i}, l^{*})^{2}\}$ and \emph{good} points as $G(S) = I \setminus B(S)$. The following lemma shows that sampling points with probability proportional to their squared lengths $\sqd{x_{i}}$ picks a bad point from $B(S)$ with probability at least $\eps/2$. %Its  proof is deferred to the Appendix. 

\begin{lemma} \label{lemma:bad-line}
If $S \subseteq [n]$ does not satisfy \eqref{eq:add-approx}, then 
$
\sum_{i \in B(S)} \sqd{x_{i}} \geq \frac{\eps}{2}~ \sum_{i=1}^{n} \sqd{x_{i}}.
$
\end{lemma}
%%%%% BELOW PROOF IS COMMENTED %%%%%

\begin{proof}
Suppose $\sum_{i \in B(S)} \sqd{x_{i}} < \eps/2~ \sum_{i=1}^{n} \sqd{x_{i}}$. Then we get a contradiction to the assumption that $S$ does not satisfy \eqref{eq:add-approx} as follows.
\begin{align*}
&\sum_{i \in I} d(x_{i}, P_{S}(l^{*}))^{2}  = \sum_{i \in G(S)} d(x_{i}, P_{S}(l^{*}))^{2} + \sum_{i \in B(S)} d(x_{i}, P_{S}(l^{*}))^{2} \\
& \leq \left(1 + \frac{\eps}{2}\right)~ \sum_{i \in G(S)} d(x_{i}, l^{*})^{2} + \sum_{i \in B(S)} \sqd{x_{i}} \\
 &\leq \left(1 + \frac{\eps}{2}\right)~ \sum_{i \in I} d(x_{i}, l^{*})^{2} + \frac{\eps}{2}~ \sum_{i=1}^{n} \sqd{x_{i}} \\
& \leq \sum_{i \in I} d(x_{i}, l^{*})^{2} + \eps~ \sum_{i=1}^{n} \sqd{x_{i}}.
\end{align*}
\end{proof}

%%%%% ABOVE PROOF IS COMMENTED %%%%%
Below we show that a \emph{bad} point sampled by squared-length sampling can be used to get another line closer to the optimal solution by a multiplicative factor, and repeat this. %A proof of the following theorem is deferred to the Appendix. 

\begin{theorem} \label{thm:additive-line}
For any given $x_{1}, x_{2}, \dotsc, x_{n} \in \R^{d}$, let $S$ be an i.i.d. sample of $O\left((1/\eps^{2})~ \log(1/\eps)\right)$ points picked by squared-length sampling. Let $I \subseteq [n]$ be the set of optimal $(1-\alpha)n$ inliers and $l^{*}$ be the optimal line that minimizes their squared distance. Then
\[
\sum_{i \in I} d(x_{i}, P_{S}(l^{*}))^{2} \leq \sum_{i \in I} d(x_{i}, l^{*})^{2} + \eps~ \sum_{i=1}^{n} \sqd{x_{i}}, \quad \text{with a constant probability}.
\]
\end{theorem}
%%%% BELOW PROOF IS COMMENTED %%%%%%
 
\begin{proof}
Lemma \ref{lemma:bad-line} shows that by squared-length sampling, the probability of picking $i \in B(S)$ is at least $\eps/2$. Now for any $i \in B(S)$, by definition we have $d(x_{i}, P_{S}(l^{*}))^{2} > (1+\eps/2)~ d(x_{i}, l^{*})^{2}$. Using this, we can show that $\linspan{S \cup \{i\}}$ has a line closer to $l^{*}$ by a multiplicative factor. That is, let $\theta_{\text{old}}$ be the angle between $l^{*}$ and $P_{S}(l^{*})$ and let $\theta_{\text{new}}$ be the angle between $l^{*}$ and $P_{S \cup \{i\}}(l^{*})$. Then $\abs{\sin{\theta_{\text{new}}}} \leq (1 - \eps/4)~ \abs{\sin \theta_{\text{old}}}$. This follows from the Angle-drop Lemma in \cite{ShyamalkumarV12} (see Lemma 13, Appendix A of \cite{DeshpandeV07}). Thus, with probability at least $\eps/2$ we pick a \emph{bad} point $x_{i}$ with $i \in B(S)$, and reduce the sine of the angle with $l^{*}$ by a multiplicative factor $(1 - \eps/4)$. We need this to happen $O((1/\eps) \log(1/\eps))$ times to bring $\abs{\sin\theta}$ down to $\eps$, and that gives an approximation with an additive error at most $\epsilon\sum_{i \in I} \norm{x_{i}}^{2} \leq \epsilon \sum_{i=1}^{n} \norm{x_{i}}^{2}$. The probability of picking a \emph{bad} point is at least $\eps/2$, so Chernoff-Hoeffding bound gives that an i.i.d. sample of size $O\left((1/\eps^{2})~ \log(1/\eps)\right)$ picked by squared-length sampling will help reduce the sine of the angle to $l^{*}$ to less than $\eps$, with a constant probability. This gives us a subset $S$ of size $\size{S} = O\left((1/\eps^{2})~ \log(1/\eps)\right)$ such that
$
\sum_{i \in I} d(x_{i}, P_{S}(l^{*}))^{2} \leq \sum_{i \in I} d(x_{i}, l^{*})^{2} + \eps~ \sum_{i=1}^{n} \sqd{x_{i}},
$ with a constant probability.
\end{proof}
 
 \subsection{Multiplicative approximation}
To turn this into a multiplicative $(1+\eps)$ guarantee, we need to use this adaptively by treating the projections of $x_{1}, x_{2}, \dotsc, x_{n}$ orthogonal to $\linspan{S}$ as our new points, and repeating the squared length sampling on these new points. Here is the modified statement of the additive approximation that we need.

\begin{theorem} \label{thm:adaptive-line}
For any given points $x_{1}, x_{2}, \dotsc, x_{n} \in \R^{d}$ and any initial subset $S_{0}$, let $S$ be an i.i.d. sample of $O\left((1/\eps^{2})~ \log(1/\eps)\right)$ points sampled with probability proportional to $d(x_{i}, \linspan{S_{0}})^{2}$. Let $I \subseteq [n]$ be the set of optimal $(1-\alpha)n$ inliers and $l^{*}$ be the optimal line that minimizes their squared distance. Then, with a constant probability, we have
\[
\sum_{i \in I} d(x_{i}, P_{S \cup S_{0}}(l^{*}))^{2} \leq \sum_{i \in I} d(x_{i}, l^{*})^{2} + \eps~ \sum_{i=1}^{n} d(x_{i}, \linspan{S_{0}})^{2}.
\]
\end{theorem}
\begin{proof}
Similar to the proof of Theorem \ref{thm:additive-line} but using the projections of $x_{i}$'s orthogonal to $\linspan{S_{0}}$ as the point set. In particular, we apply Lemma \ref{lemma:bad-line} to the projections of $x_{i}$'s orthogonal $\linspan{S_{0}}$ instead of $x_{i}$'s. Note that Theorem \ref{thm:additive-line} is a special case with $S_{0} = \emptyset$.
\end{proof}
Repeating this squared-distance sampling adaptively for multiple rounds brings the additive approximation error down exponentially in the number of rounds. %We defer a proof of the following theorem to the Appendix. 
\begin{theorem} \label{thm:adaptive-line-T}
For any given points $x_{1}, x_{2}, \dotsc, x_{n} \in \R^{d}$, any initial subset $S_{0}$ and positive integer $T$, let $S_{t}$ be an i.i.d. sample of $O\left((1/\eps^{2})~ \log(1/\eps)~ \log T\right)$ points sampled with probability proportional to $d(x_{i}, \linspan{S_{t-1}})^{2}$, for $1 \leq t \leq T$. Let $I \subseteq [n]$ be the set of optimal $(1-\alpha)n$ inliers and $l^{*}$ be the optimal line that minimizes their squared distance. Then, with a constant probability,
\[
\sum_{i \in I} d(x_{i}, P_{S_{0} \cup S_{1} \cup \dotsc \cup S_{T}}(l^{*}))^{2} \leq (1 + \eps) \sum_{i \in I} d(x_{i}, l^{*})^{2} + \eps^{T}~ \sum_{i=1}^{n} d(x_{i}, \linspan{S_{0}})^{2}. 
\]
\end{theorem}
%%%% PROOF IS COMMENTED AND DIFFERED TO APPENDIX
 
\begin{proof}
We use induction on $T$ and apply Theorem \ref{thm:adaptive-line} repeatedly. The base case $T=1$ is trivially implied by Theorem \ref{thm:adaptive-line}. 

Consider any $t \in [T]$. Applying Theorem \ref{thm:adaptive-line} to any given $S_{0} \cup S_{1} \cup \dotsc \cup S_{t-1}$ as its initial subset and picking an i.i.d. sample $S_{t}$ of $O\left((1/\eps^{2})~ \log(1/\eps)\right)$ points with probability of $x_{i}$ proportional to $d(x_{i}, \linspan{S_{0} \cup S_{1} \cup \dotsc \cup S_{t-1}})^{2}$, we get with a constant probability 
\begin{align*}
&\sum_{i \in I} d(x_{i}, P_{S_{0} \cup S_{1} \cup \dotsc \cup S_{t}}(l^{*}))^{2}\\  &\leq \sum_{i \in I} d(x_{i}, l^{*})^{2} + \eps~ \sum_{i=1}^{n} d(x_{i}, \linspan{S_{0} \cup S_{1} \cup \dotsc \cup S_{t-1}})^{2} \\
& \leq \sum_{i \in I} d(x_{i}, l^{*})^{2} + \eps~ \sum_{i=1}^{n} d(x_{i}, P_{S_{0} \cup S_{1} \cup \dotsc \cup S_{t-1}}(l^{*}))^{2}.
\end{align*}
By repeating this $O(\log T)$ times and taking the best, we can boost this success probability from a constant to $1 - 1/2T$. Since the projection of $l^{*}$ onto the span of these $O\left((1/\eps^{2})~ \log(1/\eps)~ \log T\right)$ points taken together can only be better, we get that for a larger sample $S_{t}$ of $O\left((1/\eps^{2})~ \log(1/\eps)~ \log T\right)$ points with probability of $x_{i}$ proportional to $d(x_{i}, \linspan{S_{0} \cup S_{1} \cup \dotsc \cup S_{t-1}})^{2}$, we get with at least $1 - 1/2T$ probability
\begin{align*}
&\sum_{i \in I} d(x_{i}, P_{S_{0} \cup S_{1} \cup \dotsc \cup S_{t}}(l^{*}))^{2}\\  &\leq \sum_{i \in I} d(x_{i}, l^{*})^{2} + \eps~ \sum_{i=1}^{n} d(x_{i}, \linspan{S_{0} \cup S_{1} \cup \dotsc \cup S_{t-1}})^{2} \\
& \leq \sum_{i \in I} d(x_{i}, l^{*})^{2} + \eps~ \sum_{i=1}^{n} d(x_{i}, P_{S_{0} \cup S_{1} \cup \dotsc \cup S_{t-1}}(l^{*}))^{2}.
\end{align*}
By union bound, the probability there is some $t \in [T]$ for which the above fails to hold is at most $T/2T = 1/2$. Therefore, the above holds for all $t=1, 2, \dotsc, T$ is at least $1/2$. In that case, putting these bounds together for $t=1, 2, \dotsc, T$ we get with probability at least $1/2$,
\begin{align*}
& \sum_{i \in I} d(x_{i}, P_{S_{0} \cup S_{1} \cup \dotsc \cup S_{T}}(l^{*}))^{2} \\ 
& \leq \sum_{i \in I} d(x_{i}, l^{*})^{2} + \eps~ \sum_{i=1}^{n} d(x_{i}, \linspan{S_{0} \cup S_{1} \cup \dotsc \cup S_{T-1}})^{2} \\
& \leq \sum_{i \in I} d(x_{i}, l^{*})^{2} + \eps~ \sum_{i=1}^{n} d(x_{i}, P_{S_{0} \cup S_{1} \cup \dotsc \cup S_{T-1}}(l^{*}))^{2} \\
& \leq \sum_{i \in I} d(x_{i}, l^{*})^{2} + \eps~ \left(\sum_{i \in I} d(x_{i}, l^{*})^{2} + \eps~ \sum_{i=1}^{n} d(x_{i}, P_{S_{0} \cup S_{1} \cup \dotsc \cup S_{T-2}}(l^{*}))^{2}\right) \\
& = (1 + \eps) \sum_{i \in I} d(x_{i}, l^{*})^{2} + \eps^{2}~ \sum_{i=1}^{n} d(x_{i}, P_{S_{0} \cup S_{1} \cup \dotsc \cup S_{T-2}}(l^{*}))^{2} \\
& \dotsc \\
& \leq (1 + \eps + \eps^{2} + \dotsc) \sum_{i \in I} d(x_{i}, l^{*})^{2} + \eps^{T}~ \sum_{i=1}^{n} d(x_{i}, \linspan{S_{0}})^{2} \\
& \leq \frac{1}{1-\eps}~ \sum_{i \in I} d(x_{i}, l^{*})^{2} + \eps^{T}~ \sum_{i=1}^{n} d(x_{i}, \linspan{S_{0}})^{2} \\
& \leq (1 + 2\eps) \sum_{i \in I} d(x_{i}, l^{*})^{2} + \eps^{T}~ \sum_{i=1}^{n} d(x_{i}, \linspan{S_{0}})^{2}, \qquad \text{for $\eps \leq 1/2$}.
\end{align*}
Using $\eps/2$ instead of $\eps$ in the above bound completes the proof of Theorem \ref{thm:adaptive-line-T}.

\end{proof}

Now assume that the optimal inlier error for $l^{*}$ is at least $\delta$ times its error over the entire data, that is, $\sum_{i \in I} d(x_{i}, l^{*})^{2} \geq \delta~ \sum_{i=1}^{n} d(x_{i}, l^{*})^{2}$. In that case, we can show a much stronger multiplicative $(1+\eps)$-approximation instead of additive one. %A proof of the following theorem is deferred to the Appendix. 

\begin{theorem} \label{thm:multi-line}
For any given points $x_{1}, x_{2}, \dotsc, x_{n} \in \R^{d}$, let $I \subseteq [n]$ be the set of optimal $(1-\alpha)n$ inliers and $l^{*}$ be the optimal line that minimizes their squared distance. Suppose $\sum_{i \in I} d(x_{i}, l^{*})^{2} \geq \delta~ \sum_{i=1}^{n} d(x_{i}, l^{*})^{2}$. For any $0 < \eps < 1$, we can efficiently find a subset $S$ of size  \\ $O ((1/\eps^{2}) \log(1/\eps) \log(1/\delta) \log\log(1/\delta))$ \textit{s.t.}
\[
\sum_{i \in I} d(x_{i}, P_{S}(l^{*}))^{2} \leq (1+\eps)~ \sum_{i \in I} d(x_{i}, l^{*})^{2}, \quad \text{with a constant probability}.
\]
% with a constant probability.
\end{theorem}
%%%% PROOF IS COMMENTED AND MOVED TO APPENDIX
 
\begin{proof}
We know from the squared-length sampling of Frieze et al. \cite{FriezeKV98}) that a single point $S_{0} = \{i\}$ picked with probability proportional to $\norm{x_{i}}^{2}$, gives a $2$-approximation to the best line $l$ that minimizes the sum of squared distances for all the points $x_{1}, x_{2}, \dotsc, x_{n}$, that is, $\sum_{i=1}^{n} d(x_{i}, \linspan{S_{0}})^{2} \leq 2~ \sum_{i=1}^{n} d(x_{i}, l)^{2}$. Therefore, in the presence of outliers, using the above assumption about inlier vs. total error, the same initial $S_{0}$ gives a multiplicative approximation, with a constant probability. % as follows. 
\begin{align*}
\sum_{i=1}^{n} d(x_{i}, \linspan{S_{0}})^{2} & \leq 2~ \sum_{i=1}^{n} d(x_{i}, l)^{2} 
 \leq 2~ \sum_{i=1}^{n} d(x_{i}, l^{*})^{2} 
 \leq \frac{2}{\delta}~ \sum_{i \in I} d(x_{i}, l^{*})^{2}.
\end{align*}
Now we can plug this into Theorem \ref{thm:adaptive-line-T} using $T = O\left(\log (1/\delta)\right)$ adaptive rounds to reduce the multiplicative approximation factor from $(1+\eps)/\delta$ down to a multiplicative $(1+\eps)$-approximation when compared to $\sum_{i \in I} d(x_{i}, l^{*})^{2}$. In the end, $S = S_{0} \cup S_{1} \cup \dotsc \cup S_{T}$ satisfies 
$
\sum_{i \in I} d(x_{i}, P_{S}(l^{*}))^{2} \leq (1+\eps)~ \sum_{i \in I} d(x_{i}, l^{*})^{2},
$ \text{with a constant probability}.
\end{proof}
 
%%%%%% PRROF ENDS HERE
\section{$\ell_{p}$ subspace approximation with outliers}
\label{sec:subspace}
%\vspace{-2mm}
Given an instance of $k$-dimensional subspace approximation with outliers as points $x_{1}, x_{2}, \dotsc, x_{n} \in \R^{d}$, a positive integer $1 \leq k \leq d$, a real number $p \geq 1$, and an outlier parameter $0 < \alpha < 1$, let the optimal $k$-dimensional linear subspace be $V^{*}$ that minimizes the $\ell_{p}$ error summed over its nearest $(1-\alpha)n$ points from $x_{1}, x_{2}, \dotsc, x_{n}$. Let $I \subseteq [n]$ denote the subset of indices of these nearest $(1-\alpha)n$ points to $V^{*}$. In other words, $I = N_{\alpha}(V^{*})$ consists of the indices of the optimal inliers. Given any subset $S \subseteq [n]$, let $l_{S}$ be the line or direction in it that makes the smallest angle with $V^{*}$, and define the subspace $W_{S}$ as the rotation of $V^{*}$ along this angle so as to contain $l_{S}$. To be precise, let $l^{*}$ be the projection of $l_{S}$ onto $V^{*}$ and let $W^{*}$ be the orthogonal complement of $l^{*}$ in $V^{*}$. Observe that $W_{S}$ is the $k$-dimensional linear subspace spanned by $l_{S}$ and $W^{*}$. We say that $S$ contains a line $l_{S}$ useful for \emph{additive} approximation if
\begin{align}
\sum_{i \in I} d(x_{i}, W_{S})^{p} \leq \sum_{i \in I} d(x_{i}, V^{*})^{p} + \eps~ \sum_{i=1}^{n} \norm{x_{i}}^{p}. \label{eq:add-approx-sub}
\end{align}

Define the set of \emph{bad} points as the subset of inliers $I$ whose error w.r.t. $W_{S}$ is somewhat larger than their error w.r.t. the optimal subspace $V^{*}$, that is, $B(S) = \{i \in I \suchthat d(x_{i}, W_{S})^{p} > (1+\eps/2)~ d(x_{i}, V^{*})^{p}\}$ and \emph{good} points as $G(S) = I \setminus B(S)$. The following lemma shows that sampling $i$-th points with probability proportional to $\norm{x_{i}}^{p}$ picks a \emph{bad} point $i \in B(S)$ with probability at least $\eps/2$. %A proof of the following lemma is deferred to the Appendix. 

\begin{lemma} \label{lemma:bad-sub}
If $S \subseteq [n]$ does not satisfy \eqref{eq:add-approx-sub}, then 
$
\sum_{i \in B(S)} \norm{x_{i}}^{p} \geq \frac{\eps}{2}~ \sum_{i=1}^{n} \norm{x_{i}}^{p}.
$
\end{lemma}
%%% BELOW PROOF IS COMMENTED %%%%%

\begin{proof}
Suppose $\sum_{i \in B(S)} \norm{x_{i}}^{p} < \eps/2~ \sum_{i=1}^{n} \norm{x_{i}}^{p}$. Then it implies
\begin{align*}
\sum_{i \in I} d(x_{i}, W_{S})^{p} & = \sum_{i \in G(S)} d(x_{i}, W_{S})^{p} + \sum_{i \in B(S)} d(x_{i}, W_{S})^{p} \\
& \leq \left(1 + \frac{\eps}{2}\right)~ \sum_{i \in G(S)} d(x_{i}, V^{*})^{p} + \sum_{i \in B(S)} \norm{x_{i}}^{p} \\
& \leq \left(1 + \frac{\eps}{2}\right)~ \sum_{i \in I} d(x_{i}, V^{*})^{p} + \frac{\eps}{2}~ \sum_{i=1}^{n} \norm{x_{i}}^{p} \\
& \leq \sum_{i \in I} d(x_{i}, V^{*})^{p} + \eps~ \sum_{i=1}^{n} \norm{x_{i}}^{p},
\end{align*}
which contradicts our assumption that $S$ does not satisfy \eqref{eq:add-approx-sub}.
\qedhere
\end{proof}
 
%%%% ABOVE PROOF IS COMMENTED %%%%%%%%

\subsection{Additive approximation: one dimension at a time}
Below we show that a \emph{bad} point $i \in B(S)$ can be used to improve $W_{S}$, or in other words, $\linspan{S \cup \{i\}}$ contains a line $l_{S \cup \{i\}}$ that is much closer to $V^{*}$ than $l_{S}$. %We defer a proof of the following theorem to the Appendix. 
\begin{theorem} \label{thm:additive-sub-line}
For any given points $x_{1}, x_{2}, \dotsc, x_{n} \in \R^{d}$, let $S$ be an i.i.d. sample of $O\left((p^{2}/\eps^{2})~ \log(1/\eps)\right)$ points picked with probabilities proportional to $\norm{x_{i}}^{p}$. Let $I \subseteq [n]$ be the set of optimal $(1-\alpha)n$ inliers and $V^{*}$ be the optimal subspace that minimizes the $\ell_{p}$ error over the inliers. Also let $W_{S}$ be defined as in the beginning of Section \ref{sec:subspace}. Then, with a constant probability, we have
\[
\sum_{i \in I} d(x_{i}, W_{S})^{p} \leq \sum_{i \in I} d(x_{i}, V^{*})^{p} + \eps~ \sum_{i=1}^{n} \norm{x_{i}}^{p}.
\]
\end{theorem}
%%%%% PROOF  IS COMMENTED and MOVED TO APPENDIX
 
\begin{proof}
Lemma \ref{lemma:bad-sub} shows that by sampling with probability proportional to $\norm{x_{i}}^{p}$, the probability of picking $i \in B(S)$ is at least $\eps/2$. Now for any $i \in B(S)$, we have $d(x_{i}, W_{S})^{p} > (1+\eps/2)~ d(x_{i}, V^{*})^{p}$, by definition. 

The angle-drop lemma of \cite{ShyamalkumarV12} (see Lemma 13, Appendix A of \cite{DeshpandeV07}) says that if $d(x_{i}, W_{S})^{p} > (1+\eps/2)~ d(x_{i}, V^{*})^{p}$ then $\linspan{S \cup \{i\}}$ contains a line $l_{S \cup \{i\}}$ closer to $V^{*}$ by a multiplicative factor. That is, let $\theta_{\text{old}}$ be the angle between $l_{S}$ and $V^{*}$ and let $\theta_{\text{new}}$ be the angle between $l_{S \cup \{i\}}$ and $V^{*}$, then $\abs{\sin{\theta_{\text{new}}}} \leq (1 - \eps/4p)~ \abs{\sin \theta_{\text{old}}}$. We need $O\left((p^{2}/\eps^{2}) \log(1/\eps)\right)$ such multiplicative improvements to bring $\abs{\sin \theta}$ down to $\eps^{1/p}$.

The probability of $i \in B(S)$ is at least $\eps/2$ be Lemma \ref{lemma:bad-sub}, and the sampling distribution is independent of $S$. Thus, using a careful Chernoff bound, we can show that an i.i.d. sample of $O\left((p^{2}/\eps^{2})~ \log(1/\eps)\right)$ points picked with probability proportional to $\norm{x_{i}}^{p}$ gives us, with a constant probability, enough bad points to reduce the sine of the angle between $l_{S}$ and $V^{*}$ to less than $\eps^{1/p}$. Therefore, for an i.i.d. sample $S$ of size $\size{S} = O\left((p^{2}/\eps^{2})~ \log(1/\eps)\right)$, where the $i$-th point is picked with probability proportional to $\norm{x_{i}}^{p}$, we have
\[
\sum_{i \in I} d(x_{i}, W_{S})^{p} \leq \sum_{i \in I} d(x_{i}, V^{*})^{p} + \eps~ \sum_{i=1}^{n} \norm{x_{i}}^{p}, \quad \text{with a constant probability}.
\]
% with a constant probability.
\end{proof}
 
%%%% COMMENT ENDS HERE %%%%%%%%%%%%%%%%%%
Note that we can start with any given initial subspace $S_{0}$ and prove a similar result for sampling points with probability proportional to $d(x_{i}, \linspan{S_{0}})^{p}$.
\begin{theorem} \label{thm:adaptive-sub-line}
For any given points $x_{1}, x_{2}, \dotsc, x_{n} \in \R^{d}$ and an initial subspace $S_{0}$, let $S$ be an i.i.d. sample of $O\left((p^{2}/\eps^{2})~ \log(1/\eps)\right)$ points picked with probabilities proportional to $d(x_{i}, \linspan{S_{0}})^{p}$. Let $I \subseteq [n]$ be the set of optimal $(1-\alpha)n$ inliers and $V^{*}$ be the optimal subspace that minimizes the $\ell_{p}$ error over the inliers. Also let $W_{S}$ be defined as in the beginning of Section \ref{sec:subspace}. Then, with a constant probability, we have
\[
\sum_{i \in I} d(x_{i}, W_{S \cup S_{0}})^{p} \leq \sum_{i \in I} d(x_{i}, V^{*})^{p} + \eps~ \sum_{i=1}^{n} d(x_{i}, \linspan{S_{0}})^{p}.
\]
\end{theorem}
\begin{proof}
Similar to the proof of Theorem \ref{thm:additive-sub-line} above.
\end{proof}
Once we have a line $l_{S}$ that is close to $V^{*}$, we can project orthogonal to it and repeat the sampling again. The caveat is, we do not know $l_{S}$. One can get around this by projecting all the points $x_{1}, x_{2}, \dotsc, x_{n}$ to $\linspan{S}$ of the current sample $S$, and repeat.
\begin{theorem} \label{thm:additive-sub-sub}
For any given points $x_{1}, x_{2}, \dotsc, x_{n} \in \R^{d}$, let $S = S_{1} \cup S_{2} \cup \dotsc \cup S_{k}$ be a sample of $\tilde{O}\left(p^{2}k^{2}/\eps^{2}\right)$ points picked as follows: $S_{1}$ be an i.i.d. sample of $O\left((p^{2}k/\eps^{2})~ \log(k/\eps)\right)$ points picked with probability proportional to $\norm{x_{i}}^{p}$, $S_{2}$ be an i.i.d. sample of $O\left((p^{2}k/\eps^{2})~ \log(k/\eps)\right)$ points picked with probability proportional to $d(x_{i}, \linspan{S_{1}})^{p}$, and so on. 

Let $I \subseteq [n]$ be the set of optimal $(1-\alpha)n$ inliers and $V^{*}$ be the optimal subspace that minimizes the $\ell_{p}$ error over the inliers. Then, with a constant probability, $\linspan{S}$ contains a $k$-dimensional subspace $V_{S}$ such that
\[
\sum_{i \in I} d(x_{i}, V_{S})^{p} \leq \sum_{i \in I} d(x_{i}, V^{*})^{p} + \eps~ \sum_{i=1}^{n} \norm{x_{i}}^{p}.
\]
\end{theorem}
\begin{proof}
Similar to the proof of Theorem 5, Section 4.2 in \cite{DeshpandeV07}.
\end{proof}

\subsection{Multiplicative approximation}
We can convert the above additive approximation into a multiplicative $(1+\eps)$-approximation by using this adaptively, treating the projections of $x_{1}, x_{2}, \dotsc, x_{n}$ orthogonal to $\linspan{S}$ as our new points, and repeating the sampling. 
To begin with, here is the modified statement of the additive approximation that we need.
\begin{theorem} \label{thm:adaptive-sub}
For any given points $x_{1}, x_{2}, \dotsc, x_{n} \in \R^{d}$ and any initial subset $S_{0}$, let $S = S_{0} \cup S_{1} \cup \dotsc \cup S_{k}$ be a sample of $\size{S} = \tilde{O}(p^{2}k^{2}/\eps^{2})$ points, where $S_{1}$ be an i.i.d. sample of $O\left((p^{2}k/\eps^{2})~ \log(k/\eps)\right)$
points picked with probability proportional to $d(x_{i}, \linspan{S_{0}})^{p}$, $S_{2}$ be an i.i.d. sample of \\ $O((p^{2}k/\eps^{2})\log(k/\eps))$ points picked with probability proportional to $d(x_{i}, \linspan{S_{0} \cup S_{1}})^{p}$, and so on.  Let $I \subseteq [n]$ be the set of optimal $(1-\alpha)n$ inliers and $V^{*}$ be the optimal $k$-dimensional linear subspace that minimized the $\ell_{p}$ error over the inliers. Then, with a constant probability, $S$ contains a $k$-dimensional subspace $V_{S}$ such that 
\[
\sum_{i \in I} d(x_{i}, V_{S})^{p} \leq \sum_{i \in I} d(x_{i}, V^{*})^{p} + \eps~ \sum_{i=1}^{n} d(x_{i}, \linspan{S_{0}})^{p}.
\]
\end{theorem}
\begin{proof}
Similar to the proof of Theorem \ref{thm:additive-sub-sub} but using the projections of $x_{i}$ orthogonal to $\linspan{S_{0}}$ as the point set.  Theorem \ref{thm:additive-sub-sub} is a special case with $S_{0} = \emptyset$.
\end{proof}

Repeating the result of Theorem \ref{thm:adaptive-sub} by sampling adaptively for multiple rounds brings the additive approximation error down exponentially in the number of rounds.
\begin{theorem} \label{thm:adaptive-sub-T}
For any given points $x_{1}, x_{2}, \dotsc, x_{n} \in \R^{d}$ and any initial subset $S_{0}$, let $S_{t}$ be a subset sampled by Theorem \ref{thm:adaptive-sub} after projecting the points orthogonal to $\linspan{S_{0} \cup \dotsc \cup S_{t-1}}$, for $1 \leq t \leq T$. Let $I \subseteq [n]$ be the set of optimal $(1-\alpha)n$ inliers and $V^{*}$ be the optimal $k$-dimensional linear subspace that minimizes the $\ell_{p}$ error over the inliers. Then, with a constant probability, $S = S_{0} \cup S_{1} \cup \dotsc \cup S_{T}$ of size $\size{S} = \tilde{O}\left((p^{2}k^{2}T \log T)/\eps^{2}\right)$ contains a $k$-dimensional linear subspace $V_{S}$ such that
\[
\sum_{i \in I} d(x_{i}, V_{S})^{p} \leq (1+\eps)~ \sum_{i \in I} d(x_{i}, V^{*})^{p} + \eps^{T}~ \sum_{i=1}^{n} d(x_{i}, \linspan{S_{0}})^{p}.
\]
\end{theorem}
\begin{proof}
By induction on the number of rounds $T$ and using Theorem \ref{thm:adaptive-sub}.
\end{proof} 

Now assume that the optimal $\ell_{p}$ error of $V^{*}$ over the optimal inliers $I$ is at least $\delta$ times its error over the entire data, that is, $\sum_{i \in I} d(x_{i}, V^{*})^{p} \geq \delta~ \sum_{i=1}^{n} d(x_{i}, V^{*})^{p}$. In that case, we can show a stronger multiplicative $(1+\eps)$-approximation instead of additive one. This can be thought of as a \emph{weak coreset} extending the previous work on clustering given data using points and lines \cite{FeldmanS12}.% A proof of the following theorem is deferred to the Appendix. 

\begin{theorem} \label{thm:multi-sub}
For any given points $x_{1}, x_{2}, \dotsc, x_{n} \in \R^{d}$, let $I \subseteq [n]$ be the set of optimal $(1-\alpha)n$ inliers and $V^{*}$ be the optimal $k$-dimensional linear subspace that minimizes their $\ell_{p}$ error, for $1 \leq p < \infty$. Suppose $\sum_{i \in I} d(x_{i}, V^{*})^{p} \geq \delta~ \sum_{i=1}^{n} d(x_{i}, V^{*})^{p}$. Then, for any $0 < \eps < 1$, we can efficiently find a subset $S$ of size $\size{S} = \tilde{O}\left((p^{2}k^{2}/\eps^{2})~ \log(1/\delta)~ \log\log(1/\delta)\right)$ that contains a $k$-dimensional linear subspace $V_{S}$ such that
\[
\sum_{i \in I} d(x_{i}, V_{S})^{p} \leq (1+\eps)~ \sum_{i \in I} d(x_{i}, V^{*})^{p}, \quad \text{with a constant probability}.
\]
% with a constant probability.
\end{theorem}

%%% PROOF IS MOVED TO APPENDIX %%%
 
\begin{proof}
We know from Theorem 3 of \cite{DeshpandeV07} that using approximate volume sampling, we can efficiently find a subset $S_{0}$ of size $\size{S_{0}} = k$ such that $V_{S_{0}} = \linspan{S_{0}}$ gives a multiplicative $2^{O(pk \log k)}$ to the optimal subspace $V$ that minimizes the $\ell_{p}$ error over all the points $x_{1}, x_{2}, \dotsc, x_{n}$, that is, 
$
\sum_{i=1} d(x_{i}, V_{S_{0}})^{p} \leq 2^{O(pk \log k)}~ \sum_{i=1}^{n} d(x_{i}, V)^{p}.
$
Therefore, in the presence of outliers, using the above assumption about inlier vs. total error, $V_{S_{0}} = \linspan{S_{0}}$ gives a $2^{O(pk \log k)} \cdot 1/\delta$ multiplicative approximation to the $\ell_{p}$ subspace approximation problem with outliers as follows. 
% \begin{align*}
% \sum_{i=1}^{n} d(x_{i}, V_{S_{0}})^{p} & \leq 2^{O(pk \log k)}~ \sum_{i=1}^{n} d(x_{i}, V)^{p} \\
% & \leq 2^{O(pk \log k)}~ \sum_{i=1}^{n} d(x_{i}, V^{*})^{p} \\
% & \leq \frac{2^{O(pk \log k)}}{\delta}~ \sum_{i \in I} d(x_{i}, V^{*})^{p}.
% \end{align*}
\begin{align*}
\sum_{i=1}^{n} d(x_{i}, V_{S_{0}})^{p} & \leq 2^{O(pk \log k)}~ \sum_{i=1}^{n} d(x_{i}, V)^{p}\\ 
&\leq 2^{O(pk \log k)}~ \sum_{i=1}^{n} d(x_{i}, V^{*})^{p} \\
& \leq \frac{2^{O(pk \log k)}}{\delta}~ \sum_{i \in I} d(x_{i}, V^{*})^{p}.
\end{align*}
Now we can plug this into Theorem \ref{thm:adaptive-sub-T} using $T = O\left(pk \log k \log (1/\delta) \log (1/\eps)\right)$ adaptive rounds to reduce the multiplicative approximation factor from $(1+\eps)/\delta$ down to a multiplicative $(1+\eps)$-approximation when compared to $\sum_{i \in I} d(x_{i}, V^{*})^{p}$. Putting it all together, we get a subset $S$ of size $\size{S} = \tilde{O}\left(p^{3}k^{3}/\eps^{2} \cdot \log(1/\delta)\right)$ whose span contains a $k$-dimensional linear subspace $V_{S}$ such that
$$
\sum_{i \in I} d(x_{i}, V_{S})^{p} \leq (1+\eps)~ \sum_{i \in I} d(x_{i}, V^{*})^{p}.
$$
\end{proof}
 
 \section{M-estimator subspace approximation with outliers}
$\ell_{p}$ error or loss function is a special case of M-estimators used in statistics. General M-estimators as loss functions for subspace approximation or clustering have been previously studied in \cite{FeldmanS12} for point and line median clustering and in \cite{ClarksonW15} for robust regression. One way to define robust variants of the subspace approximation problem is to use more general loss functions that are more resilient to outliers. Here are a few examples of popular M-estimators.
\begin{itemize}
\item Huber's loss function with threshold parameter $t$
\[
L(x) = \begin{cases} x^{2}/2, & \text{if $\abs{x} < t$} \\ t \abs{x} - t^{2}/2, & \text{if $\abs{x} \geq t$}. \end{cases}
\]
\item Tukey's biweight or bisquare loss function with threshold parameter $t$
\[
L(x) = \begin{cases} \left(t^{6} - (t^{2} - x^{2})^{3}\right)/6, & \text{if $\abs{x} < t$} \\ 0, & \text{if $\abs{x} \geq t$}. \end{cases}
\]
% \item Andrew's loss function
% \[
% L(x) = \begin{cases} 1 - \cos(x), & \text{if $\abs{x} \leq \pi$} \\ 0, & \text{if $\abs{x} > \pi$}. \end{cases}
% \]
\end{itemize}
The advantage of more general loss functions such as Huber loss is that they approximate squared-error for the nearer points but approximate $\ell_{1}$-error for faraway points. They combine the smoothness of squared-error with the robustness of $\ell_{1}$-error.

Clarkson and Woodruff \cite{ClarksonW15} study the M-estimator variant of subspace approximation defined as follows. Given points $x_{1}, x_{2}, \dotsc, x_{n} \in \R^{d}$, an integer $1 \leq k \leq d$, and an M-estimator loss function $M: \R \rightarrow \R$, find a $k$-dimensional linear subspace $V$ that minimizes
$
\sum_{i=1}^{n} M\left(d(x_{i}, V)\right).
$
Clarkson and Woodruff \cite{ClarksonW15} show that the adaptive sampling for angle-drop lemma used by \cite{DeshpandeV07} to go from a large multiplicative approximation down to $(1+\eps)$-approximation can also be achieved by a non-adaptive \emph{residual} sampling. Here we restate Theorem 45 from \cite{ClarksonW15} using our notation of subspaces and distances instead of matrix norms.
\begin{theorem}{(Theorem 45 of \cite{ClarksonW15})}
Given $x_{1}, x_{2}, \dotsc, x_{n} \in \R^{d}$, an integer $1 \leq k \leq d$, and an M-estimator loss function $M(\cdot)$, let $V_{0}$ be any linear subspace such that $\sum_{i=1}^{n} M(d(x_{i}, V_{0})) \leq C~ \sum_{i=1}^{n} M(d(x_{i}, V))$, where $V$ is the $k$-dimensional linear subspace that minimizes the M-estimator error for its distances to $x_{1},\dotsc, x_{n}$ summed over all the points. Let $S \subseteq [n]$ be a sample of points, where each $i$ gets picked independently with probability $\min\{1, C' \cdot M(d(x_{i}, V_{0}))/\sum_{i=1}^{n} M(d(x_{i}, V_{0})\}$, for some constant $C' = O\left(C k^{3}/\eps^{2}~ \log(k/\eps)\right)$. Then, with a constant probability, we have
\begin{itemize}
    \item $\sum_{i=1}^{n} M(d(x_{i}, \linspan{V_{0} \cup S})) \leq (1+\eps)~ \sum_{i=1}^{n} M(d(x_{i}, V))$, and
    \item $\size{S} = O\left(C k^{3}/\eps^{2}~ \log(k/\eps)\right)$.
\end{itemize}
%\[
%\sum_{i=1}^{n} M(d(x_{i}, \linspan{V_{0} \cup S})) \leq (1+\eps)~ \sum_{i=1}^{n} %M(d(x_{i}, V)), \text{and~} \size{S} = O\left(C k^{3}/\eps^{2}~ %\log(k/\eps)\right).
%\]
% and $\size{S} = O\left(C k^{3}/\eps^{2}~ \log(k/\eps)\right)$.
\end{theorem} 
The advantage of their algorithm is that it does not require multiple passes to do the adaptive sampling, and moreover, it can be combined with an approximate residual score computation to finally get an algorithm that runs in time linear in the number of non-negative coordinates in the input data.
 
For general loss functions or M-estimators, one can define an analogous variant of the subspace approximation problem with outliers as follows. Given points $x_{1}, x_{2}, \dotsc, x_{n} \in \R^{d}$, an integer $1 \leq k \leq d$, a monotone M-estimator loss function $M: \R_{\geq 0} \rightarrow \R_{\geq 0}$, and an outlier parameter $0 \leq \alpha \leq 1$, find a $k$-dimensional linear subspace $V$ that minimizes the sum of M-estimator loss of distances to the nearest $(1-\alpha)n$ points. In other words, let $N_{\alpha}(V) \subseteq [n]$ consist of the indices of the nearest $(1-\alpha)n$ points to $V$ among $x_{1}, x_{2}, \dotsc, x_{n}$. We want to find a $k$-dimensional linear subspace $V$ that minimizes
$
\sum_{i \in N_{\alpha}(V)} M\left(d(x_{i}, V)\right).
$
This variant allows us to control the robustness in two ways: explicitly, using the outlier parameter in the definition, and implicitly, using an appropriate M-estimator loss function of our choice.

We observe that the proof of Theorem 45 in \cite{ClarksonW15} is based on angle-drop lemma and our arguments in Section \ref{sec:subspace} for subspace approximation with outliers go through with very little or no change. Thus, we have the following theorem similar to their dimension reduction for subspace approximation, whose   proof  is similar to the proofs of Theorems 41 and 45 in \cite{ClarksonW15}.

\begin{theorem} \label{thm:m-est}
For any given points $x_{1}, x_{2}, \dotsc, x_{n} \in \R^{d}$, let $I \subseteq [n]$ be the set of optimal $(1-\alpha)n$ inliers and $V^{*}$ be the optimal $k$-dimensional linear subspace that minimizes their M-estimator error $\sum_{i \in I} M(d(x_{i}, V^{*}))$. Suppose $\sum_{i \in I} d(x_{i}, V^{*})^{p} \geq \delta~ \sum_{i=1}^{n} d(x_{i}, V^{*})^{p}$. Then, for any $0 < \eps < 1$, we can efficiently find a subspace $V'$ of dimension $\tilde{O}\left(p^{2}k^{2}/\eps^{2})~ \log(1/\delta) \log\log(1/\delta)\right)$ such that, with a constant probability, it contains a $k$-dimensional linear subspace $\tilde{V}$ satisfying
\[
\sum_{i \in I} d(x_{i}, \tilde{V})^{p} \leq (1+\eps)~ \sum_{i \in I} d(x_{i}, V^{*})^{p}.
\]
\end{theorem}
% \begin{proof}
% Similar to the proofs of Theorems 41 and 45 in \cite{ClarksonW15}.
% \end{proof}
 
\section{Affine subspace approximation with outliers} 
Given an input data set in a high-dimensional space, \emph{affine subspace approximation} asks for an affine subspace that best fits this data. For squared-error or $\ell_{2}$ subspace approximation, it is easy to see that the best such subspace must pass through the mean of the input data. An easy way to see this is using the parallel axis theorem.

\begin{prop} \label{prop:parallel}
Given any points $x_{1}, x_{2}, \dotsc, x_{n} \in \R^{d}$ and an affine subspace $V$, let $V_{\mu}$ be the parallel translate of $V$ that passes through the mean $\mu = \sum_{i=1}^{n} x_{i}/n$. Then
\[
\sum_{i=1}^{n} d(x_{i}, V)^{2} = \sum_{i=1}^{n} d(x_{i}, V_{\mu})^{2} + n~ d(V, V_{\mu})^{2},
\]
where $d(V, V_{\mu}) = \min\{d(x, y) \suchthat x \in V~ \text{and}~ y \in V_{\mu}\}$.
\end{prop}

Consider the problem of affine subspace approximation with the squared error in the presence of outliers as follows. Given points $x_{1}, x_{2}, \dotsc, x_{n} \in \R^{d}$, an integer $1 \leq k \leq d$, and an outlier parameter $0 \leq \alpha \leq 1$, find a $k$-dimensional affine subspace $V$ that minimizes the sum of squared distances of the $(1-\alpha)n$ points nearest to it. In other words, let $N_{\alpha}(V) \subseteq [n]$ consist of the indices of the nearest $(1-\alpha)n$ points to $V$ among $x_{1}, x_{2}, \dotsc, x_{n}$. We want to minimize
\[
\sum_{i \in N_{\alpha}(V)} d(x_{i}, V)^{2}.
\]
In other words, the subspace approximation problem with outliers, for a given outlier parameter $0 \leq \alpha \leq 1$, is to consider all partitions of $x_{1}, x_{2}, \dotsc, x_{n}$ into $(1-\alpha)n$ inliers and the remaining $\alpha n$ outliers, and find the affine subspace with the least squared error for the inliers over all such partitions.

Let $V^{*}$ be the optimal solution to the above problem, and let $I = N_{\alpha}(V^{*})$ be the optimal set of inliers. By Proposition \ref{prop:parallel}, $V^{*}$ must pass through the mean of the inliers, that is, through $\mu^{*} = \sum_{i \in I} x_{i}/(1-\alpha)n$. If we could sample points from $I$ uniformly at random, then $\mu^{*}$ can be well-approximated by the empirical mean of a small sample. Since we do not know $I$ but know that $\size{I} = (1-\alpha)n$, we can use a trick that is often used in $k$-means clustering and related problems \cite{KumarSS2004}. We can pick a small, uniformly random sample of points from $[n]$, then go over all its partitions into two parts by brute force, and go over the empirical means of the two parts for each partition. One of these partitions will correspond to the correct inlier-outlier partition of our sample. In that case, the mean of the inlier part behaves like the empirical mean of a uniformly random sample of inliers.

We first state a lemma, which is implicit in Theorem 2 from a paper of Barman \cite{Barman2015} on approximate Caratheodory's theorem. We reproduce its short proof for completeness.
\begin{lemma} \label{lemma:sample-mean}
Let $S$ be a i.i.d. random sample of $2/\eta^{2}$ points from $I$ picked uniformly. Let $\mu_{S}$ be the mean of the sample $S$ and $\mu_{I}$ be the mean of all the points in $I$. Then
\[
\prob{\norm{\mu_{S} - \mu_{I}} \leq \eta D} \geq 1/2.
\]
%\[
%\prob{\norm{\frac{1}{\size{S}} \sum_{i \in S} x_{i} - \frac{1}{\size{I}} \sum_{i \in I} x_{i}} \leq \delta D} \geq 1/2,
%\]
where $D$ is the diameter of the set $\{x_{i} \suchthat i \in I\}$.
\end{lemma}
\begin{proof}
Let $r_{1}, r_{2}, \dotsc, r_{\size{S}}$ be the i.i.d. uniform random points from $\{x_{i} \suchthat i \in I\}$.
\begin{align*}
\expec{\norm{\mu_{S} - \mu_{I}}^{2}} & = \frac{1}{\size{S}^{2}}~ \expec{\norm{\sum_{i=1}^{\size{S}} (r_{i} - \mu_{I})}^{2}} \\
& = \frac{1}{\size{S}^{2}}~ \sum_{i=1}^{\size{S}} \sum_{j=1}^{\size{S}} \expec{\inner{r_{i} - \mu_{I}}{r_{j} - \mu_{I}}} \\
& = \frac{1}{\size{S}^{2}}~ \sum_{i=1}^{\size{S}} \expec{\norm{r_{i} - \mu_{I}}^{2}} \qquad \text{as $r_{i}, r_{j}$ independent, and $\expec{r_{i}} = \mu_{I}$} \\
& \leq \frac{D^{2}}{\size{S}}, \qquad \text{where $D$ is the diameter of $\{x_{i} \suchthat i \in I\}$}.
\end{align*}
Therefore,
\[
\prob{\norm{\mu_{S} - \mu_{I}} > \sqrt{2} D/\sqrt{\size{S}}} = \prob{\norm{\mu_{S} - \mu_{I}}^{2} > 2 D^{2}/\size{S}} \leq 1/2.
\]
In other words, if $\size{S} = 2/\eta^{2}$ then $\prob{\norm{\mu_{S} - \mu_{I}} \leq \eta D} \geq 1/2$.
\end{proof}

\begin{theorem}
For any given points $x_{1}, x_{2}, \dotsc, x_{n} \in \R^{d}$, let $I \subseteq [n]$ be the set of optimal $(1-\alpha)n$ inliers and $V^{*}$ be the optimal $k$-dimensional affine subspace that minimizes their squared distance. Suppose $\sum_{i \in I} d(x_{i}, V^{*})^{2} \geq \delta~ \sum_{i=1}^{n} d(x_{i}, V^{*})^{2}$. Then, for any $0 < \eps < 1$, we can efficiently find, in time linear in $n$ and $d$, a $k$-dimensional linear subspace $V'$ such that
\[
\sum_{i \in I} d(x_{i}, V')^{2} \leq \sum_{i \in I} d(x_{i}, V^{*})^{2} + \eps~ \sum_{i \in I} \norm{x_{i}}^{2},
\]
with a constant probability.
\end{theorem}
\begin{proof}
We pick a sample $T$ of $O\left(1/\eta^{2} (1-\alpha)\right)$ points uniformly at random from $x_{1}, x_{2}, \dotsc, x_{n}$, and then go over all partitions of this sample $T$ into two parts $(S, T \setminus S)$ by brute force. We consider the means $\mu_{S} = \sum_{i \in S} x_{i}/\size{S}$ for the part $S$, shift our entire data as $x_{1} - \mu_{S}, x_{2} - \mu_{S}, \dotsc, x_{n} - \mu_{S}$, and solve the subspace approximation problem for linear subspaces on this shifted input. By Lemma \ref{lemma:sample-mean}, we know that $\norm{\mu_{S} - \mu_{I}} \leq \eta D$, with a constant probability. The optimal affine subspace $V^{*}$ passes through the mean $\mu_{I}$ of the inliers. However, by Proposition \ref{prop:parallel}, these exists a parallel translate of $V^{*}$ through $\mu_{S}$, call it $\tilde{V}$, such that
\[
\sum_{i \in I} d(x_{i}, \tilde{V})^{2} \leq \sum_{i \in I} d(x_{i}, V^{*})^{2} + (1-\alpha) n~ \eta^{2} D^{2}.
\]
Thus, if we get a $(1+\eps)$-approximation to the shifted instance $x_{1} - \mu_{S}, x_{2} - \mu_{S}, \dotsc, x_{n} - \mu_{S}$ for linear subspace approximation using Theorem \ref{thm:multi-sub}, we essentially get an affine subspace approximation with squared error at most
\[
(1+\eps)~ \sum_{i \in I} d(x_{i}, V^{*})^{2} + (1+\eps)~ (1-\alpha) n~ \eta^{2} D^{2},
\]
with a constant probability. Looking carefully through the proof of Lemma \ref{lemma:sample-mean}, the guarantee is actually at most
\[
(1+\eps)~ \sum_{i \in I} d(x_{i}, V^{*})^{2} + (1+\eps)~ \eta^{2}~ \sum_{i \in I} \norm{x_{i} - \mu_{I}}^{2},
\]
which is at most $(1+\eps)~ \sum_{i \in I} d(x_{i}, V^{*})^{2} + \eps~ \sum_{i \in I} \norm{x_{i}}^{2}$ or an additive $2\eps~ \sum_{i \in I} \norm{x_{i}}^{2}$, for an appropriate choice of $\eta = O(\sqrt{\eps})$.
\end{proof}

\bibliographystyle{plain}
\bibliography{reference}

\begin{thebibliography}{10}

\bibitem{Barman2015}
Siddharth Barman.
\newblock Approximating nash equilibria and dense bipartite subgraphs via an
  approximate version of caratheodory's theorem.
\newblock In {\em Proceedings of the Forty-Seventh Annual {ACM} on Symposium on
  Theory of Computing, {STOC} 2015, Portland, OR, USA, June 14-17, 2015}, pages
  361--369, 2015.

\bibitem{BK2018}
Aditya Bhaskara and Srivatsan Kumar.
\newblock Low rank approximation in the presence of outliers.
\newblock In {\em Approximation, Randomization, and Combinatorial Optimization.
  Algorithms and Techniques, {APPROX/RANDOM} 2018, August 20-22, 2018 -
  Princeton, NJ, {USA}}, pages 4:1--4:16, 2018.

\bibitem{Chen2008}
Ke~Chen.
\newblock A constant factor approximation algorithm for \emph{k}-median
  clustering with outliers.
\newblock In {\em Proceedings of the Nineteenth Annual {ACM-SIAM} Symposium on
  Discrete Algorithms, {SODA} 2008, San Francisco, California, USA, January
  20-22, 2008}, pages 826--835, 2008.

\bibitem{ClarksonW15}
Kenneth~L. Clarkson and David~P. Woodruff.
\newblock Input sparsity and hardness for robust subspace approximation.
\newblock In {\em {IEEE} 56th Annual Symposium on Foundations of Computer
  Science, {FOCS} 2015, Berkeley, CA, USA, 17-20 October, 2015}, pages
  310--329, 2015.

\bibitem{DeshpandeV07}
Amit Deshpande and Kasturi~R. Varadarajan.
\newblock Sampling-based dimension reduction for subspace approximation.
\newblock In {\em Proceedings of the 39th Annual {ACM} Symposium on Theory of
  Computing, San Diego, California, USA, June 11-13, 2007}, pages 641--650,
  2007.

\bibitem{FeldmanL11}
Dan Feldman and Michael Langberg.
\newblock A unified framework for approximating and clustering data.
\newblock In {\em Proceedings of the 43rd {ACM} Symposium on Theory of
  Computing, {STOC} 2011, San Jose, CA, USA, 6-8 June 2011}, pages 569--578,
  2011.

\bibitem{FeldmanMSW10}
Dan Feldman, Morteza Monemizadeh, Christian Sohler, and David~P. Woodruff.
\newblock Coresets and sketches for high dimensional subspace approximation
  problems.
\newblock In {\em Proceedings of the Twenty-First Annual {ACM-SIAM} Symposium
  on Discrete Algorithms, {SODA} 2010, Austin, Texas, USA, January 17-19,
  2010}, pages 630--649, 2010.

\bibitem{FeldmanS12}
Dan Feldman and Leonard~J. Schulman.
\newblock Data reduction for weighted and outlier-resistant clustering.
\newblock In {\em Proceedings of the Twenty-Third Annual {ACM-SIAM} Symposium
  on Discrete Algorithms, {SODA} 2012, Kyoto, Japan, January 17-19, 2012},
  pages 1343--1354, 2012.

\bibitem{FriezeKV98}
Alan~M. Frieze, Ravi Kannan, and Santosh~S. Vempala.
\newblock Fast monte-carlo algorithms for finding low-rank approximations.
\newblock {\em J. {ACM}}, 51(6):1025--1041, 2004.

\bibitem{GhashamiP14}
Mina Ghashami and Jeff~M. Phillips.
\newblock Relative errors for deterministic low-rank matrix approximations.
\newblock In {\em Proceedings of the Twenty-Fifth Annual {ACM-SIAM} Symposium
  on Discrete Algorithms, {SODA} 2014, Portland, Oregon, USA, January 5-7,
  2014}, pages 707--717, 2014.

\bibitem{HardtM13}
Moritz Hardt and Ankur Moitra.
\newblock Algorithms and hardness for robust subspace recovery.
\newblock In {\em {COLT} 2013 - The 26th Annual Conference on Learning Theory,
  June 12-14, 2013, Princeton University, NJ, {USA}}, pages 354--375, 2013.

\bibitem{Khachiyan1995}
Leonid Khachiyan.
\newblock On the complexity of approximating extremal determinants in matrices.
\newblock {\em J. Complex.}, 11(1):138--153, 1995.

\bibitem{KKM2018}
Adam~R. Klivans, Pravesh~K. Kothari, and Raghu Meka.
\newblock Efficient algorithms for outlier-robust regression.
\newblock In {\em Conference On Learning Theory, {COLT} 2018, Stockholm,
  Sweden, 6-9 July 2018}, pages 1420--1430, 2018.

\bibitem{KrishnaswamyLS2018}
Ravishankar Krishnaswamy, Shi Li, and Sai Sandeep.
\newblock Constant approximation for k-median and k-means with outliers via
  iterative rounding.
\newblock In {\em Proceedings of the 50th Annual {ACM} {SIGACT} Symposium on
  Theory of Computing, {STOC} 2018, Los Angeles, CA, USA, June 25-29, 2018},
  pages 646--659, 2018.

\bibitem{KumarSS2004}
Amit Kumar, Yogish Sabharwal, and Sandeep Sen.
\newblock A simple linear time $(1+\epsilon)$ -approximation algorithm for
  k-means clustering in any dimensions.
\newblock In {\em Proceedings of the 45th Annual IEEE Symposium on Foundations
  of Computer Science}, FOCS '04, pages 454--462, Washington, DC, USA, 2004.
  IEEE Computer Society.

\bibitem{LRV2016}
Kevin~A. Lai, Anup~B. Rao, and Santosh~S. Vempala.
\newblock Agnostic estimation of mean and covariance.
\newblock In {\em {IEEE} 57th Annual Symposium on Foundations of Computer
  Science, {FOCS} 2016, 9-11 October 2016, Hyatt Regency, New Brunswick, New
  Jersey, {USA}}, pages 665--674, 2016.

\bibitem{Liberty13}
Edo Liberty.
\newblock Simple and deterministic matrix sketching.
\newblock In {\em The 19th {ACM} {SIGKDD} International Conference on Knowledge
  Discovery and Data Mining, {KDD} 2013, Chicago, IL, USA, August 11-14, 2013},
  pages 581--588, 2013.

\bibitem{RS2010}
Prasad Raghavendra and David Steurer.
\newblock Graph expansion and the unique games conjecture.
\newblock In {\em Proceedings of the 42nd {ACM} Symposium on Theory of
  Computing, {STOC} 2010, Cambridge, Massachusetts, USA, 5-8 June 2010}, pages
  755--764, 2010.

\bibitem{Sarlos06}
Tam{\'{a}}s Sarl{\'{o}}s.
\newblock Improved approximation algorithms for large matrices via random
  projections.
\newblock In {\em 47th Annual {IEEE} Symposium on Foundations of Computer
  Science {(FOCS} 2006), 21-24 October 2006, Berkeley, California, USA,
  Proceedings}, pages 143--152, 2006.

\bibitem{ShyamalkumarV12}
Nariankadu~D. Shyamalkumar and Kasturi~R. Varadarajan.
\newblock Efficient subspace approximation algorithms.
\newblock {\em Discrete {\&} Computational Geometry}, 47(1):44--63, 2012.

\bibitem{Woodruff-monograph}
David~P. Woodruff.
\newblock Sketching as a tool for numerical linear algebra.
\newblock {\em Foundations and Trends® in Theoretical Computer Science},
  10(1–2):1--157, 2014.

\end{thebibliography}

\end{document}